\documentclass[conference]{IEEEtran}%
\usepackage{amsfonts}
\usepackage{amsmath}
\usepackage{amssymb}
\usepackage{graphicx}%
\providecommand{\U}[1]{\protect\rule{.1in}{.1in}}
\newtheorem{theorem}{Theorem}

\newtheorem{corollary}{Corollary}

\newtheorem{definition}{Definition}

\newtheorem{remark}{Remark}

\begin{document}

\title{Smart Meter Privacy:\ A Utility-Privacy Framework}
\pubid{~}
\specialpapernotice{~}

%

\author{\authorblockN{S. Raj Rajagopalan\authorrefmark{1},
Lalitha Sankar\authorrefmark{2},  Soheil Mohajer\authorrefmark{2}%
, H. Vincent Poor\authorrefmark{2}}
\authorblockA{\authorrefmark{1}HP Labs,
Princeton, NJ 08544.
raj.raj@hp.com\\}
\authorblockA{\authorrefmark{2}Dept. of Electrical Engineering,
Princeton University,
Princeton, NJ 08544.
{lalitha,smohajer,poor}@princeton.edu\\}
}%
%

\maketitle
%

\begin{abstract}%

\footnotetext{The research was supported in part by the National Science
Foundation under Grants CCF-$10$-$16671$ and CNS-$09$-$05398$, in part by the
Air Force Office of Scientific Research MURI Grant FA-$9550$-$09$-$1$-$0643$,
and in part by DTRA under Grant HDTRA1-07-1-0037.}End-user privacy in smart
meter measurements is a well-known challenge in the smart grid. The solutions
offered thus far have been tied to specific technologies such as batteries or
assumptions on data usage. Existing solutions have also not quantified the
loss of benefit (utility) that results from any such privacy-preserving
approach. Using tools from information theory, a new framework is presented
that abstracts both the privacy and the utility requirements of smart meter
data. This leads to a novel privacy-utility tradeoff problem with minimal
assumptions that is tractable. Specifically for a stationary Gaussian Markov
model of the electricity load, it is shown that the optimal
utility-and-privacy preserving solution requires filtering out frequency
components that are low in power, and this approach appears to encompass most
of the proposed privacy approaches.%

\end{abstract}%

\section{Introduction}

Information collection and dissemination, some of it using smart meters, are
critical to the smart grid. But information about electricity consumption that
is collected and harnessed for a more efficient and multi-faceted grid may be
used for purposes beyond electricity consumption, thereby making it
potentially dangerous to individual privacy. The privacy consequences of smart
grid development are hard to understand for two principal reasons:\ (i)\ the
\ full range of technological capabilities and information extraction
possibilities have not been laid out, and (ii)\ our concept of privacy in this
space are yet poorly defined and shifting. Smart meters are an indispensable
enabler in the context of smart grids, which deploy advanced information and
communication technology to control the electrical grid.

The main motivations for high-resolution energy usage data collection are to
forecast load demand and to provide optimized service to consumers in the form
of pricing structure \cite{QuinnColState}. An electricity provider can use
this information to facilitate more efficient network management, peak load
reduction, load shaping, and a number of other such uses. However, it has been
known for some time that the information of appliance use can be reconstructed
from the overall real-time load using libraries of appliance load signatures
that could be matched to signals found within the noise of a customer's
aggregated electricity use and a large amount of detail concerning customer
usage habits can be discerned \cite{SGPArt2}\emph{. }\cite{QuinnColState}
cites a list of privacy-sensitive characteristics that may be inferred from
electricity load data ranging from house occupancy to personal habits and routines.

The NIST\ Smart Grid Interoperability Panel has also underlined risk to
privacy of personal behavior because new types of energy use data are created
and communicated by smart meters, such as unique electric signatures for
consumer electronics and appliances, thereby opening up further opportunities
for general invasion of privacy. \cite{Cavoukian-2010} suggest that there will
always \textquotedblleft be the temptation to sell such information such as
energy usage or appliance data, either in identifiable customer level,
anonymized or aggregate form to third parties such as marketers seeking
commercial gain.\textquotedblright. Thus, a desired feature of privacy design
in the smart grid would be \textquotedblleft positive-sum, not
zero-sum\textquotedblright\ in that it seeks to accommodate all legitimate
interests and objectives in a fair manner without completely sacrificing
privacy for utility or vice-versa.

A typical approach to privacy in smart meter data is aggregation along
dimensions of space (using neighborhood gateways, e.g. \cite{Li-SmartGrid2010}%
), time (using battery storage, e.g. \cite{SGPProc2}), or precision (by noise
addition, e.g. \cite{Rastogi:2010}). These solutions seek to support utility
and privacy in different ways; however, they do not have a robust theoretical
basis for both privacy and utility. Such a basis is important for several
reasons. First, a theoretical abstraction allows us to recast the problem in a
technology-independent manner -- we need a privacy framework that not only
addresses the capabilities of current non-intrusive load monitoring (NALM)
techniques but is also extensible to future ones. Second, a theoretical
framework enables us to examine the costs of lost privacy against the benefits
of data dissemination, namely, the tradeoff between privacy and utility. It
would be desirable to give each customer the ability to decide that tradeoff
and also to give the electricity provider the ability to incentivize the
customer to participate in such a bargain by offering interesting points of
tradeoff. Finally, a theoretical framework for privacy and utility may expose
points of tradeoff that are unexpected. 

We propose a general theoretical framework that brings most current treatments
of the privacy-utility tradeoff into a single model -- it enables us to look
at a spectrum of abstract privacy-utility choices and enables us to find
maximal points on such a tradeoff curve. It also suggests new possible ways of
achieving this tradeoff that have not been considered thus far.

What we have found is that suppressing low power components would be
consistent with intuitive notions of privacy in smart meter data. At the same
time, our utility constraints guarantee that the bulk of the energy
consumption information in the load measurements is retained in the revealed
data. This suggests that it may indeed be possible to reveal significant
energy consumption information without also revealing a lot of personal
information and the resulting tradeoff can be tuned. This would be an
interesting avenue for further exploration. The paper is organized as follows.
In Section \ref{SecII}, we outline current approaches to smart meter privacy.
In Section \ref{SecIII}, we develop our model, metrics, and the
privacy-utility tradeoff framework and illustrate our results in Section
\ref{SecIV}.

\section{\label{SecII}Related Work}

The advantages and usefulness of smart meters in general is examined in a
number of papers; see for example \cite{Deconinck} and the references therein.
\cite{SGPProc2} presents a pioneering view of privacy of smart meter
information:\ the authors identify the need for privacy in a home's load
signature as being an inference violation (resulting from load signatures of
home appliances)\ rather than an identity violation (i.e. loss of anonymity).
Accordingly, they propose home electrical power routing using rechargeable
batteries and alternate power sources to moderate the effects of load
signatures. They also propose three different privacy metrics: relative
entropy, clustering classification, and a correlation/regression metric.
However they do not propose any formal utility metrics to quantify the
utility-privacy tradeoff.

Recently, \cite{SGPProc1} proposes additional protection through the use of a
trusted escrow service, along with randomized time intervals between the setup
of attributable and anonymous data profiles at the smart meter.
\cite{Molina-Markham:2010} shows, somewhat surprisingly, that even without
\textit{a priori} knowledge of household activities or prior training it is
possible to extract complex usage patterns from smart meter data such as
residential occupancy and social activities very accurately using
off-the-shelf statistical methods. \cite{Li-SmartGrid2010}\ and
\cite{Molina-Markham:2010} propose privacy-enhancing designs using
neighborhood-level aggregation and cryptographic protocols to communicate with
the energy supplier without compromising the privacy of individual homes.
However, escrow services and neighborhood gateways support only restricted
query types and do not completely solve the problem of trustworthiness.
\cite{Varodayan-ICASSP2011} presents a formal state transition diagram-based
analysis of the privacy afforded by the rechargeable battery model proposed in
\cite{SGPProc2}. However, \cite{Varodayan-ICASSP2011} does not offer a
comparable model of utility to compare the risks of information leakage with
the benefits of the information transmitted.

In, \cite{Rastogi:2010} the authors present a method of providing
\emph{differential privacy} over aggregate queries modeling smart meter
measurements as time-series data from multiple sources containing temporal
correlations. While their approach has some similarity to ours in terms of
time-series data treatment, their method does not seem generalizable to
arbitrary query types. On the other hand, \cite{Papadimitriou:2007} introduces
the notion of partial information hiding by introducing uncertainty about
individual values in a time series by perturbing them. Our method is a more
general approach to time series data perturbation that guarantees that the
perturbation cannot be eliminated by averaging.

\section{\label{SecIII}Our Contributions}

The primary challenge in characterizing the privacy-utility tradeoffs for
smart meter data is creating the right abstraction -- we need a principled
approach that provides quantitative measures of both the amount of information
leaked as well as the utility retained, does not rely on any assumptions of
data mining algorithms, and provides a basis for a negotiated level of benefit
for both consumer and supplier \cite{Cavoukian-2010}.
\cite{Varodayan-ICASSP2011} provides the beginnings of such a model -- they
assume that in every sampling time instant, the net load is either 0 or 1
power unit represented by the smart meter readings $X_{k}$, $k=1,2,...,$ are a
discrete-time sequence of binary independent and identically distributed
values. They model the battery-based filter of \cite{SGPProc2} as a stochastic
transfer function that outputs a binary sequence $\hat{X}_{k}$ that tells the
electricity provider whether the home is drawing power or not at any given
moment. The amount of information leaked by the transfer function is defined
to be the mutual information rate $I(X;\hat{X})$\ between the random variables
$X$ and $\hat{X}$. By modeling the battery charging policy as a 2-state
stochastic transition machine, they show that there exist battery policies
that result in less information leakage than from the deterministic charging
policy of \cite{SGPProc2}. Though \cite{Varodayan-ICASSP2011} does not provide
a general utility function to go with the chosen privacy function and the
modeling assumptions are extremely simplistic, it nevertheless provides a good
starting point for our framework.

In our model, we assume that the load measurements are sampled (at an
appropriate frequency)\ from a smart meter, that they are real-valued, and can
be correlated (models the temporal memory of both appliances and human usage
patterns). Rather than assume any specific transfer function, we assume an
abstract transfer function which maps the input load measurements $X$ into an
output sequence $\hat{X}$. As in \cite{Varodayan-ICASSP2011}, we assume a
mutual information rate as a metric for privacy leakage; however, we allow for
the fact that a large space of (unknown to us) inferences can be made from the
meter data -- we model the inferred data as a random variable $Y$ correlated
with the measurement variable $X$. Thus, the privacy leakage is the mutual
information between $Y$ and $\hat{X}$. We also provide an abstract utility
function which measures the fidelity of the output sequence $\hat{X}$ by
limiting the Euclidean distance (mean square error) between $X$ and $\hat{X}$.
Using these abstractions and tools from the theory of rate distortion we are
able to meet all our requirements for a general but tractable privacy-utility
framework:\ the privacy and utility requirements provide opposing constraints
that expose a spectrum of choices for trading off privacy for utility and vice-versa.

\subsection{Model}

We write $x_{t}$, $t=1,2,\ldots,n,$ to denote the sampled load measurements
from a smart meter. In general$,$ $x_{t}$ are complex valued corresponding to
the real and reactive measurements and are typically vectors for multi-phase
systems \cite{SGPArt2}. For simplicity and ease of presentation, we model the
meter measurements as a sequence of real-valued scalars (for example, such a
model applies to two-phase $120~$V appliances for which one of the two phase
components is zero).

For appropriately small sampling intervals, the smart meter time-series data
that result from sampling the underlying continuous-time continuous-amplitude
processes can be viewed as being generated by a random source with memory. The
memory models the continuity and the effect of both short-term and long-term
correlations in the load measurements. The short term correlations typically
model the effect of the set of appliances in use over the said duration while
the long term correlations model the long term power usage pattern of the
human user. We model the continuous valued smart meter data as a sequence
$\ldots,X_{k-1},X_{k},X_{k+1},\ldots$, of random variables $X_{k}%
\in\mathcal{X}$, $-\infty<k<\infty,$ generated by a stationary continuous
valued source with memory. Specifically, we model the continuous valued
discrete-time smart meter data as a sequence $\ldots,X_{k-1},X_{k}%
,X_{k+1},\ldots$, of Gaussian random variables $X_{k}\in\mathcal{X}$,
$k=0,\pm1,...$, generated by a stationary Gaussian source with memory captured
via the autocorrelation function
\begin{equation}
c_{XX}\left(  m\right)  =E\left[  X_{k}X_{k+m}\right]  ,m=0,\pm1,\pm
2,....\label{corr_func}%
\end{equation}
The assumption of normal distribution for total load is a simplification from
empirical observations \cite{HP} that the power consumption pattern of a
typical appliance in the on state is approximately Gaussian.

\subsection{Utility and Privacy Metrics}

Since continuous amplitude sources cannot be transmitted losslessly over
finite capacity links, a sampled sequence of $n$ load measurements $X^{n}$ is
compressed before transmission. In general, however, even if the sampled
measurements were quantized \textit{a priori}, i.e., take values in a discrete
alphabet, there may be a need to perturb (distort) the data in some way to
guarantee a measure of privacy. However, such a perturbation also needs to
maintain a desired level of fidelity.

Intuitively, utility of the perturbed data is high if any function computed on
it yields results similar to those from the original data; thus, the utility
is highest when there is no perturbation and goes to zero when the perturbed
data is completely unrelated to the original. Accordingly, our utility metric
is an appropriately chosen average `distance' \textit{distortion function}
between the original and the perturbed data.

Privacy, on the other hand, is maximized when the perturbed data is completely
independent of the original. Our privacy metric measures the difficulty of
inferring any private information of the data collector's choice, defined as a
sequence $\left\{  Y_{k}\right\}  $ of random variables $Y_{k}\in\mathcal{Y}$,
-$\infty<k<\infty$, which is correlated with and can be inferred from the
revealed data. The random sequence $\left\{  Y_{k}\right\}  $ for all $k$
along with the joint distribution $p_{X^{n}Y^{n}}$ mathematically captures the
space of all inferences that can be made from the measurements. We quantify
the resulting privacy loss as a result of revealing perturbed data via the
\textit{mutual information} between the two data sequences.

As an aside, we note here that our model of privacy is between a single user
(household)\ and the electricity provider. It does not consider the leakage
possibilities of comparing the perturbed data from two or more different
users. On the other hand it can address the possibility of side-information
such as income level of the user that may cause further information leakage.
If we know the statistics of the side-information that we can incorporate the
possible leakage into the model and derive the consequent modified
privacy-utility tradeoff. For simplicity we ignore the side-information aspect
in this paper.

\subsection{Perturbation: Encoding and Decoding}

\textit{Encoding}: We assume that a meter collects $n\gg1$ measurements in an
interval of time prior to communication and that $n$ is large enough to
capture the source's memory. The encoding function is then a mapping of the
resulting \textit{source sequence} $X^{n}=\left(  X_{1}\text{ }X_{2}\text{
}X_{3}\text{ }\ldots\text{ }X_{n}\right)  $, where $X_{k}\in\mathbb{R},$ for
all $k=1,2,\ldots,n,$ to an index $W_{n}\in\mathcal{W}_{n}$ given by%
\begin{equation}
F_{E}:\mathcal{X}^{n}\rightarrow\mathcal{W}_{n}\equiv\left\{  1,2,\ldots
,M_{n}\right\}  \label{F_Enc}%
\end{equation}
where each index represents a quantized sequence.

\textit{Decoding}: The decoder (at the data collector) computes an output
sequence $\hat{X}^{n}=\left(  \hat{X}_{1}\text{ }\hat{X}_{2}\text{ }\hat
{X}_{3}\text{ }\ldots\text{ }\hat{X}_{n}\right)  ,$ $\hat{X}_{k}\in\mathbb{R}%
$, for all $k,$ using the decoding function%
\begin{equation}
F_{D}:\mathcal{W}\rightarrow\mathcal{\hat{X}}^{n}.\label{F_Dec}%
\end{equation}
The encoder is chosen such that the input and output sequences achieve a
desired utility given by an average distortion constraint
\begin{equation}
D_{n}=\frac{1}{n}%
{\textstyle\sum\limits_{k=1}^{n}}
\mathbb{E}\left[  \left(  X_{k}-\hat{X}_{k}\right)  ^{2}\right]
\label{Distortion}%
\end{equation}
and a constraint on the information leakage about the desired sequence
$\left\{  Y_{k}\right\}  $ from the revealed sequence $\left\{  \hat{X}%
_{k}\right\}  $ is quantified via the leakage function
\begin{equation}
L_{n}=\frac{1}{n}I\left(  Y^{n};\hat{X}^{n}\right)  \label{Leakage}%
\end{equation}
where $\mathbb{E}\left[  \cdot\right]  $ denotes the expectation over the
joint distribution of $X^{n}$ and $\hat{X}^{n}$ given by $p_{X\hat{X}}\left(
x^{n},\hat{x}^{n}\right)  =P_{X^{n}}^{n}\left(  x^{n}\right)  p_{t}(\hat
{x}^{n}|x^{n})$ where $p_{t}(\hat{x}^{n}|x^{n})$ is a conditional pdf on
$\hat{x}^{n}$ given $x^{n}$. The mean-square error (MSE) distortion function
chosen in (\ref{Distortion}) is typical for Gaussian distributed real-valued
data as a measure of the fidelity of the perturbation (encoding). 

Note that $D_{n}$ and $L_{n}$ are functions of the number of measurements $n$
and for stationary sources converge to limiting values \cite{CTbook}. Let $D$
and $L$ denote the corresponding limiting values for utility and privacy,
respectively, i.e.,
\begin{equation}%
\begin{array}
[c]{ccc}%
D\equiv\lim\limits_{n\rightarrow\infty}D_{n} & \text{and} & L\equiv
\lim\limits_{n\rightarrow\infty}L_{n}.
\end{array}
\end{equation}

\subsection{Utility-Privacy Tradeoff Region}

Formally, the utility-privacy tradeoff region $\mathcal{T}$ is defined as follows.

\begin{definition}
The smart meter utility-privacy tradeoff region $\mathcal{T}$ is the set of
all $\left(  D,L\right)  $ pairs for which there exists a coding scheme given
by (\ref{F_Enc}) and (\ref{F_Dec}) with parameters $(n,M_{n},D_{n}%
+\epsilon,L_{n}+\epsilon)$ satisfying (\ref{Distortion}) and (\ref{Leakage})
for $n$ sufficiently large and $\epsilon>0$.
\end{definition}

\textit{Rate-Distortion-Leakage}: The above utility-privacy tradeoff problem
does not explicitly bound the number $M_{n}$ of encoded
(quantized)\ sequences. An explicit constraint on
\begin{equation}
M_{n}\leq2^{n(R_{n}+\epsilon)} \label{Rate}%
\end{equation}
results in a rate-distortion-leakage (RDL) tradeoff problem for which the
feasible region is defined as follows. Let $R=\lim_{n\rightarrow\infty}\left(
\log M_{n}\right)  /n.$

\begin{definition}
The rate-distortion-leakage tradeoff region $\mathcal{R}_{RDL}$ is the set of
all $\left(  R,D,L\right)  $ tuples for which there exists a coding scheme
given by (\ref{F_Enc}), (\ref{F_Dec}), and (\ref{Rate}) with parameters
$(n,M_{n},D_{n}+\epsilon,L_{n}+\epsilon)$ satisfying (\ref{Distortion}) and
(\ref{Leakage}) for $n$ sufficiently large and $\epsilon>0$. The function
$\lambda\left(  D\right)  $ quantifies the minimal leakage achievable for a
feasible distortion $D$ such that the set of all $\left(  R,D,\lambda
(D)\right)  $ are the boundary points of $\mathcal{R}_{RDL}$.
\end{definition}

\begin{theorem}
\label{Th_RDEEq}$\mathcal{T=}\left\{  \left(  D,L\right)  :(R,D,L)\in
\mathcal{R}_{RDL},D\in\lbrack0,\right.  $ $\left.  D_{\max}],L\geq
\lambda\left(  D\right)  \right\}  .$
\end{theorem}

\textit{Proof sketch}: The crux of our argument is the fact that for any
feasible utility vector $D$, choosing the minimum rate $R\left(
D,\lambda(D)\right)  $, ensures that the least amount of \textit{information}
is revealed about the source via the reconstructed variable. This in turn
ensures that the minimal leakage $\lambda(D)$ of the correlated sequence
$Y^{n}$ is achieved for that utility. For the same utility constraint, since
such a rate requirement is not a part of the utility-privacy model, the
resulting maximal privacy achieved is at most as large as that in
$\mathcal{R}_{RDL}$.

\subsection{Rate-Distortion-Leakage Tradeoff}

We now use Theorem \ref{Th_RDEEq} to precisely quantify the utility-privacy
tradeoff via the RDL tradeoff region. The proof is a direct generalization of
the RDL region for memoryless sources (see, for example, \cite{Yamamoto,SRV4}%
), and hence, is omitted for lack of space. Intuitively, the proof follows
from upper and lower bounding the minimal communication rate $R$ as a function
of $D$ and $L$ and the minimal leakage rate $\lambda$ as a function of $D$.

\begin{theorem}
The rate-distortion-leakage region for a source with memory subject to
distortion and leakage constraints in (\ref{Distortion}) and (\ref{Leakage})
is given by the rate-distortion and minimal leakage functions%
\begin{align}
R(D,L)  &  =\lim_{n\rightarrow\infty}\inf_{p\left(  x^{n},y^{n}\right)
p\left(  \hat{x}^{n}|x^{n}\right)  }\frac{1}{n}I\left(  X^{n};\hat{X}%
^{n}\right) \label{RDL_n}\\
\lambda(D)  &  =\lim_{n\rightarrow\infty}\inf_{p\left(  x^{n},y^{n}\right)
p\left(  \hat{x}^{n}|x^{n}\right)  }\frac{1}{n}I\left(  Y^{n};\hat{X}%
^{n}\right)  . \label{lambdaD_n}%
\end{align}
The utility-privacy tradeoff is captured by $\lambda(D)$ which is the minimal
privacy leakage for a desired distortion (utility) $D$.
\end{theorem}

\begin{remark}
The Markov relationship $Y^{n}-X^{n}-\hat{X}^{n}$ is captured via the set of
all distributions in (\ref{RDL_n}) and (\ref{lambdaD_n}) which minimize
$R(D,L)$ and $\lambda(D)$.
\end{remark}

\begin{corollary}
For $Y_{k}=X_{k}$, for all $k,$ i.e., for the case in which the actual
measurements need to be undisclosed, $\lambda(D)=R(D,L)=R(D)$ where $R(D)$ is
the rate-distortion function for the source.
\end{corollary}

In general, the optimal distribution minimizing the rate subject to both the
distortion and leakage constraints depends on the joint distribution of the
measurement and inference sequences. Modeling this relationship is, in
general, not straightforward or known \textit{a priori}. Given this
limitation, we consider a simple linear inference model given by
\begin{equation}
Y_{k}=\alpha_{k}X_{k}+Z_{k},\text{ for all }k,\label{YkXk_model}%
\end{equation}
where $Z_{k}\sim\mathcal{N}\left(  0,1\right)  $ is independent of $X_{k}$,
and $\alpha_{k}$ are constants. In this paper, we limit our results to these
models to simplify our analysis and develop the intuition that can eventually
lead us to develop complete solutions for a more general inference model. The
following theorem captures our result.

\begin{theorem}
\label{Th_RDE2RD}The utility-privacy tradeoff for smart meter measurements
modeled as a Gaussian source with memory with $Y_{k}=\alpha_{k}X_{k}+Z_{k}$,
for all $k$, is given by the leakage function $\lambda(D)$ which results from
choosing the distribution $p\left(  \hat{x}^{n}|x^{n}\right)  $ as the
rate-distortion (without privacy) optimal distribution.
\end{theorem}

\begin{proof}
The proof follows directly from noting that, for a given jointly Gaussian
distribution of the source and correlated hidden sequence, $p_{X^{n}Y^{n}}$,
the infimum in (\ref{RDL_n}) and (\ref{lambdaD_n}) is strictly over the space
of conditional distributions of the revealed sequence given the original
source sequence as a result of the Markov chain relationship $Y^{n}-X^{n}%
-\hat{X}^{n}.$ Expanding the leakage as $I(Y^{n};\hat{X}^{n})=h(Y^{n}%
)-h(Y^{n}|\hat{X}^{n}),$ and using the fact for correlated Gaussian processes,
$Y_{k}=\alpha_{k}X_{k}+Z_{k}$, for all $k$, where $\left\{  Z_{k}\right\}  $
is a sequence independent of $\left\{  X_{k}\right\}  $ and $\alpha_{k}$ is a
constant for each $k$, one can show that the jointly Gaussian distribution of
$X^{n}$ and $\hat{X}^{n}$ which minimizes (\ref{RDL_n}) also minimizes
(\ref{lambdaD_n}).
\end{proof}

\begin{remark}
Theorem \ref{Th_RDE2RD} simplifies the development of the RDL region for
Gaussian sources with memory for which the rate-distortion function is known.
For Gaussian sources with memory the rate-distortion function is known and
lends itself to a straightforward practical implementation that we discuss in
the following section.
\end{remark}

\subsection{Rate-Distortion for Gaussian Sources with\ Memory}

In general, the rate distortion functions for sources with memory are not
straightforward to compute. However, for Gaussian sources, the rate-distortion
function $R(D)$ (without the additional privacy constraint) is known and can
be obtained via a transformation of the correlated source sequence $X^{n}$ to
its eigen-space in which the transformed sequence $\tilde{X}^{n}$ is a
collection of independent random variables with, in general, different variances.

A standard approach to analyze correlated data is to project the data to an
orthogonal basis in which the leakage and distortion constraints remain
invariant. Since the data is random, we project on to the principal axes of
the $n\times n$ correlation matrix $G_{XX}$ whose entry in the $i^{th}$ row
and $j^{th}$ column is $c\left(  i-j\right)  $ defined in (\ref{corr_func})
for which the mean-square error (Euclidean distance) function and the mutual
information leakage are invariant. Thus, while the constraints for the
original and transformed measurements are the same, the advantage of the
transformation is that the resulting measurements in any block of length $n$
are statistical independent.

We write $S_{X}\left(  f\right)  $ denote the unitary transformation of the
correlation matrix $G_{XX}$, i.e., $S_{X}\left(  f\right)  $ is the power
spectral density (PSD) of the time series process $\left\{  X\left(  n\right)
\right\}  $, at discrete frequencies, $f=0,1,2,...,n-1.$ We henceforth refer
to the transform domain as the spectral domain in keeping with the literature.
Similarly, let $\mathcal{S}_{Y}(\omega)$ and $\mathcal{S}_{XY}(\omega)$ denote
the PSDs of the $\left\{  Y_{k}\right\}  $ and the $\left\{  X_{k}%
Y_{k}\right\}  $ processes where $\mathcal{S}_{XY}(\omega)$ is the transform
of the cross-correlation function $c_{XY}\left(  m\right)  $ of the two
sequences. Let $\phi$ denote the Lagrangian parameter for the distortion
constraint (\ref{Distortion}) in the rate minimization problem. Explicitly
denoting the dependence on the water-level $\phi$, the rate-distortion
function $R_{\phi}\left(  D\right)  $ and the average distortion function
$D\left(  \phi\right)  $ are given by \cite{DSS15}%
\begin{align}
R_{\phi}\left(  D\right)   &  =\int_{-\pi}^{\pi}\max\left(  0,\frac{1}{2}%
\log\frac{\mathcal{S}_{X}(\omega)}{\phi}\right)  \frac{d\omega}{2\pi
}\label{RDGauss}\\
D\left(  \phi\right)   &  =\int_{-\pi}^{\pi}\min\left(  \mathcal{S}_{X}%
(\omega),\phi\right)  \frac{d\omega}{2\pi}.\label{DphiGauss}%
\end{align}
Note that the water-level $\phi$ is determined by the desired average
distortion $D\left(  \phi\right)  =D$. Thus, $R(D)$ for a Gaussian source with
memory can be expressed as an infinite sum of the rate-distortion functions
for independent Gaussian variables, one for each angular frequency $\omega
\in\lbrack-\pi,\pi].$ The \textquotedblleft water-level\textquotedblright%
\ $\phi$ captures the average time-domain distortion constraint across the
spectrum such that the distortion for any $\omega$ is the minimum of the
water-level and the PSD. The privacy leakage $\lambda(D\left(  \phi\right)  )$
is then the infinite sum of the information leakage about $\left\{
Y_{k}\right\}  $ for each $\omega,$ and is given by%
\begin{equation}
\lambda\left(  D\left(  \phi\right)  \right)  =\int_{-\pi}^{\pi}\frac{1}%
{2}\log\left(  \frac{\mathcal{S}_{Y}(\omega)}{\mathcal{S}_{XY}(\omega)g\left(
\omega\right)  +\mathcal{S}_{Y}(\omega)}\right)  \frac{d\omega}{2\pi}%
\end{equation}
where $g\left(  \omega\right)  \equiv\left(  \min\left(  \mathcal{S}%
_{X}(\omega),\phi\right)  -1\right)  .$%
\begin{figure}
[ptb]
\begin{center}
\includegraphics[
trim=0.493711in 0.246865in 0.494854in 0.247580in,
height=2.2667in,
width=3.5397in
]%
{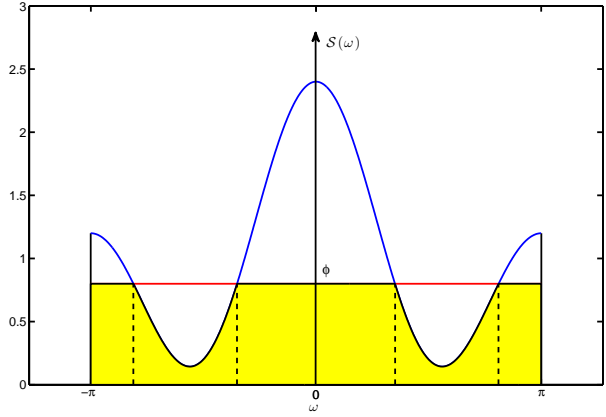}%
\caption{The PSD of $\{X_{k}\}$. The area below the curve and the horizontal
line is equal to $D$.}%
\label{FigPSD}%
\end{center}
\end{figure}

\begin{remark}
The transform domain \textquotedblleft waterfilling\textquotedblright%
\ solution suggests that in practice the time-series data can be filtered for
a desired level of fidelity (distortion) and privacy (leakage) using Fourier
transforms. The privacy-preserving rate-distortion optimal scheme thus reveals
only those frequency components with power above the water-level $\phi$.
Furthermore, at every frequency only the portion of the signal energy which is
above the water level $\phi$ is preserved by the minimum-rate sequence from
which the source can be generated with an average distortion $D$.
\end{remark}

\section{\label{SecIV}Illustration}

The following example illustrates our results. We assume that the private
information to be hidden is the measurement sequence itself, i.e.,
$Y_{k}=X_{k},$ for all $k.$ For the meter measurements modeled as a stationary
Gaussian time series $\left\{  X_{k}\right\}  ,$ we choose $X_{k}%
\sim\mathcal{N}(0,1)$ for all $k\in\mathcal{I}$, and an autocorrelation
function
\[
c_{m}=\mathbb{E}[X_{k}X_{k+m}]=\left\{
\begin{array}
[c]{ll}%
1 & m=0,\\
0.3 & m=\pm1,\\
0.4 & m=\pm2,\\
0 & \mathrm{otherwise.}%
\end{array}
\right.
\]
The power spectral density PSD (frequency domain representation of the
autocorrelation function) of this process is given by
\begin{multline}
\mathcal{S}(\omega)=\sum_{m=-\infty}^{\infty}c_{m}\exp(im\omega)=1+0.6\cos
(\omega)+0.8\cos(2\omega),\\
-\pi\leq\omega\leq\pi.
\end{multline}
In order to obtain the rate-distortion function $R_{\phi}(D)$ for this source,
for a given $D$ we have to find the water-level $\phi$ satisfying
(\ref{DphiGauss}).

Figure~\ref{FigPSD} shows the PSD function $c_{m}$. Determining $\phi$ is
equivalent to determining the height of the horizontal line, such that the
area below the curve and the line equals $D$ as given by (\ref{DphiGauss}).
Having determined $\phi$, $R_{\phi}(D)$ is then given by (\ref{RDGauss}).
$\mathcal{S}(\omega)$ takes its minimum value at $\omega_{0}=\arccos(\frac
{-3}{16})\simeq1.7594$. Thus, for $D\leq\mathcal{S}(\omega_{0})\simeq0.1437$,
$\phi=D$ such that $R(D)=\frac{1}{4\pi}\int_{-\pi}^{\pi}\log\left(
\mathcal{S}(\omega)/D\right)  d\omega,$which is the same as the
rate-distortion function for a Gaussian source with variance $\tilde{\sigma
}^{2}=\frac{1}{2\pi}\int_{-\pi}^{\pi}\log\mathcal{S}(\omega)d\omega$, i.e.,
when the distortion falls below a certain threshold, the rate required to
reproduce the source at the receiver with the desired fidelity is the same as
that of a memoryless Gaussian source. Finally, since we have chosen to hide
the original meter measurements, for this problem, the privacy leakage is the
same as the rate distortion. The resulting tradeoff between is shown in Fig.
\ref{FigRD}.%

\begin{figure}
[ptb]
\begin{center}
\includegraphics[
trim=0.497466in 0.259822in 0.542536in 0.261705in,
height=2.418in,
width=3.1298in
]%
{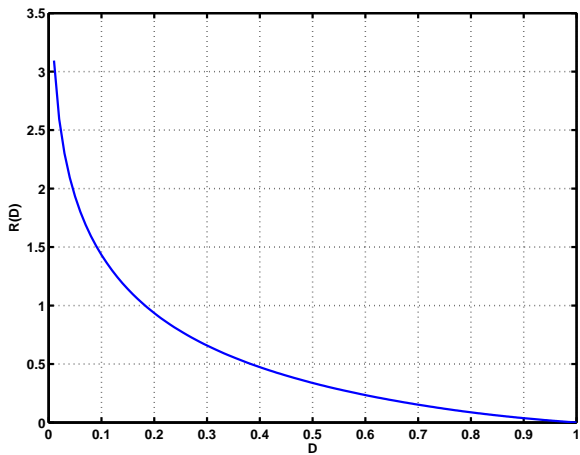}%
\caption{Plot of $R_{\phi}(D)$ $=$ $\lambda_{\phi}(D)$ vs. average distortion
$D.$}%
\label{FigRD}%
\end{center}
\end{figure}

\section{\label{SecV}Discussion and Concluding Remarks}

The theoretical framework that we have developed here allows us to precisely
quantify the utility-privacy tradeoff problem in smart meter data. Given a
series of smart meter measurements $X$, we reveal a perturbation $\hat{X}$
that allows us to guarantee a measure of both privacy in $X$\ and utility in
$\hat{X}$. The privacy guarantee comes from the bound on information leakage
while the utility guarantee comes from the upper bound on the MSE\ distance
between $X$\ and $\hat{X}$.

Our model of privacy, namely information leakage, does not depend on any
assumptions about the inference mechanism (i.e. the data mining algorithms);
instead it presents the least possible (on average) guarantee of information
leakage about $X$, while the utility is preserved in an application-agnostic
manner. Our framework is also agnostic about how the perturbation is achieved;
for example it can be achieved using a filter such as a battery or by adding noise.

Modeling a smart meter as a Gaussian source with memory and extending known
results from rate distortion theory, we show that a utility-privacy tradeoff
framework can be constructed that gives tight bounds on the amount of privacy
that can be achieved for a given level of utility and vice-versa. The critical
parameter of choice in the utility-privacy tradeoff is the water level $\phi$,
which in turn depends on the bound on the distortion that is acceptable. The
choice of $\phi$ dictates the extent to which the original signal (meter
measurements)\ can be distorted and the rate $R_{\phi}\left(  D\right)  $\ is
the maximum data precision allowed for which information leakage is at most
$\lambda(D\left(  \phi\right)  )$. In a practical context, the choice of
$\phi$ is dictated by the choice of the privacy-utility tradeoff operating
point, which in turn has to be negotiated between the energy provider and consumer.

Our distortion model can be viewed as a filter on the load signal $X$ -- it
filters out all frequencies that have power below a certain threshold
(determined directly by $\phi$). This filter is novel and comes directly as a
result of our model. From a practical point of view, it makes sense in the
following way. From the appliance signature chart in \cite{QuinnColState},
frequency components that have low power typically correspond to fluctuations
in energy consumption that are short-lived, which in turn are caused by
appliances such as kettles and television sets and transmit the bulk of
information about underlying human behavior. Frequency components that have
high power tend to caused by continuously running appliances such as air
conditioning units and refrigerators that reveal much less about human
behavior. Suppressing low power components would thus reduce or eliminate the
components of the signal that are likely to be most revealing about human
behavior and thus match our intuition on privacy protection in smart meter
data. At the same time, our utility constraints guarantee that most of the
useful energy consumption information is retained in the revealed load data.
This holds out hope that we can reveal significant energy consumption
information while at the same time protecting significant personal information
in a tunable tradeoff. This would be an interesting avenue for further
exploration. Another interesting avenue to explore would be to apply and
demonstrate the power of these concepts in a practical context.

\bibliographystyle{IEEEtran}
\bibliography{refsSG}

\end{document}